\documentclass[11pt,a4paper,twoside]{article}
  \usepackage{newapa}
	\usepackage{tikz}
  \usepackage{a4wide}
	\usepackage{graphicx} 
	\usepackage{wrapfig}
	\usepackage[small,bf,hang]{caption}
	\usepackage{amssymb}  
	\usepackage{amsmath}  
	\usepackage{amsthm}
	\usepackage[english]{babel} 
	\usepackage{dsfont}	
	\DeclareMathOperator{\A}{\mathcal{A}} 
	
	\DeclareMathOperator{\pee}{\mathds{P}}																													
	\DeclareMathAlphabet{\mathpzc}{OT1}{pzc}{m}{it} 	
  \theoremstyle{plain}
	\newtheorem{theorem}{Theorem}
\begin{document}

\title{Speakable in Quantum Mechanics}

\author{Ronnie Hermens\thanks{Department of Theoretical Philosophy, University of Groningen, Oude Bo\-te\-ringe\-straat 52, 9712 GL Gro\-ning\-en, Netherlands}}

\maketitle

\begin{abstract}
At the 1927 Como conference Bohr spoke the now famous words ``It is wrong to think that the task of physics is to find out how nature is. Physics concerns what we can say about nature.''
However, if the Copenhagen interpretation really holds on to this motto, why then is there this feeling of conflict when comparing it with realist interpretations?
Surely what one \emph{can} say about nature should in a certain sense be interpretation independent.
In this paper I take Bohr's motto seriously and develop a quantum logic that avoids assuming any form of realism as much as possible.
To illustrate the non-triviality of this motto a similar result is first derived for classical mechanics.
It turns out that the logic for classical mechanics is a special case of the derived quantum logic.
Finally, some hints are provided in how these logics are to be used in practical situations and I discuss how some realist interpretations relate to these logics. 
\end{abstract}

\section{Introduction}
\label{intro}
Over the last few decades much of research in the foundations of quantum mechanics has focused on the impossibility of certain interpretations. 
Results such as the Kochen-Specker theorem \cite{KS67} or the Bell inequalities \cite{Bell64}, \cite{Clauser69} mainly establish what is unspeakable in quantum mechanics.
Often these results are interpreted in favor of instrumentalist or Copenhagen-like interpretations of quantum mechanics.
But as it is well known, switching to an epistemic account of physics alone isn't sufficient to account for the counter-intuitive aspects of quantum mechanics.
Furthermore, these accounts often resort to vagueness when explaining, for example, how quantum mechanics can violate Bell inequalities.
A notorious example is Bohr's account of complementarity.

In this paper I propose a formal reasoning scheme for epistemic approaches in physics in order to shed some light on what is speakable in quantum mechanics.
After all, the more clearly one understands the empirical part of quantum mechanics, the easier the ontological part can be investigated.
To get a feeling for the problems I have in mind, consider the following result from probability logic. 
\begin{theorem}\label{Belllemma}
Suppose $\pee$ is a probability function on a collection of sentences $S$ that satisfies the following rules for all $A,B\in S$:
\begin{enumerate}
\item If $A\vdash B$, then $\pee(A)\leq\pee(B)$.
\item $\pee(A\vee B)\leq\pee(A)+\pee(B)$.
\end{enumerate}
Then, if $S$ obeys classical logic, the following inequality holds for all $A_1,A_2,B_1$ and $B_2$ in $S$:
\begin{equation}\label{Bellineq}
	\pee(A_1\wedge B_1)\leq
	\pee(A_1\wedge B_2)+\pee(A_2\wedge B_1)+\pee(\neg A_2\wedge \neg B_2).
\end{equation}
\end{theorem}
\begin{proof}
The result follows by writing out in the following way:
\begin{equation}
\begin{split}
	\pee(A_1\wedge B_1)
	&=
	\pee(A_1\wedge B_1\wedge(B_2\vee\neg B_2))=\pee((A_1\wedge B_1\wedge B_2)\vee(A_1\wedge B_1\wedge\neg B_2))\\
	&\leq
	\pee(A_1\wedge B_1\wedge B_2)+\pee(A_1\wedge B_1\wedge\neg B_2)\leq\pee(A_1\wedge B_2)+\pee(B_1\wedge\neg B_2)\\
	&=
	\pee(A_1\wedge B_2)+\pee(B_1\wedge\neg B_2\wedge(A_2\vee\neg A_2))\\
	&=
	\pee(A_1\wedge B_2)+\pee((B_1\wedge\neg B_2\wedge A_2)\vee(B_1\wedge\neg B_2\wedge\neg A_2))\\
	&\leq
	\pee(A_1\wedge B_2)+\pee(B_1\wedge\neg B_2\wedge A_2)+\pee(B_1\wedge\neg B_2\wedge\neg A_2)\\
	&\leq
	\pee(A_1\wedge B_2)+\pee(A_2\wedge B_1)+\pee(\neg A_2\wedge \neg B_2).
\end{split}
\end{equation}	
\end{proof} 

It is well-known that quantum mechanics is capable of violating such inequalities. 
To see this one can consider a pair of entangled qubits and set $A_i=[\sigma^A_{r_i}=\tfrac{1}{2}]$, $\neg A_i=[\sigma^A_{r_i}=-\tfrac{1}{2}]$, $B_i=[\sigma^B_{r_i}=\tfrac{1}{2}]$ and $\neg B_i=[\sigma^B_{r_i}=-\tfrac{1}{2}]$ for $i=1,2$, where $\sigma^A_{r_i}$ is the spin along the $r_i$-axis of the qubit send to Alice, and $\sigma^B_{r_i}$ the spin along the $r_i$-axis of the qubit send to Bob.
Now any interpretation (realist or instrumentalist) of quantum mechanics must be able to point out a flaw in this theorem.
For example, non-local theories may argue that a revelation of the truth value of, say, $A_1$ instantaneously causes an altering of the truth values of $B_1$ and $B_2$.
From a Copenhagen perspective one may argue for example that $B_1$ and $B_2$ are complementary sentences which cannot meaningfully occur both in a single sentence, making the proof of theorem \ref{Belllemma} meaningless.
A similar approach would be used by followers of the consistent histories approach or the many worlds interpretation.
And the list can go on.

One of the responses that I find most intriguing is that from orthodox quantum logic \cite{BirkhoffNeumann36}.
In this approach the proof simply fails because it uses the law of distributivity several times; the premises of the theorem do not hold in quantum mechanics. 
The elegance of this approach is that it points to a flaw in the proof precisely there where the proof clashes with quantum mechanical calculations.
But as far as explanations go, this approach only replaces a mystery with another mystery as long as no explanation is given for the failure of distributivity.
In particular it raises questions about the meaning of the logical connectives ``and'' and ``or''. 
It is well argued that their meaning in quantum logic should at least differ from the classical meaning \cite{Dummett76}, but it isn't that clear what it should be instead.
From this perspective, traditional quantum logic has failed in providing a framework for reasoning about quantum mechanical phenomena.

A striking example (due to Popper) of the problem with interpreting logical connectives in quantum logic is the following. 
The law of excluded middle is maintained in quantum logic and therefore, for every proposition $P$ the formula $P\vee\neg P$ is always true. 
Furthermore, for every pair of propositions $P_1$ and $P_2$ it holds that if $P_1$ is true, then $P_1\wedge(P_2\vee\neg P_2)$ is also true. 
However, in quantum logic there are such pairs for which $P_1$ is true, but neither $P_1\wedge P_2$ nor $P_1\wedge\neg P_2$ is true (a failure of the law of distributivity). 
That is, $P_1$ presents itself as a proposition that is both incompatible with $P_2$ and $\neg P_2$ and may thus be thought of as an excluded middle. 
Obviously, this contradiction arose because I held on to a certain interpretation of the logical connectives. 
But the derivation seems innocent enough for me to conclude that no satisfactory interpretation of the quantum logical connectives can be defined. 
Rather, I tend to agree with Popper that
\begin{quote}	``the kind of change in classical logic which would fit what Birkhoff and von Neumann suggest [\ldots] would be the rejection of the law of excluded middle [\ldots], as proposed by Brouwer, but rejected by Birkhoff and von Neumann'' \cite{Popper68}.
\end{quote}

But a crude shift to intuitionistic logic in which the law of excluded middle is simply thrown overboard seems unsatisfactory.
In some cases one simply finds it to be true. 
More specifically, upon a measurement of $\sigma^A_{r_i}$ one knows that either $A_i$ or $\neg A_i$ will be true.
The decidability of a proposition thus depends on the context in which the proposition is formulated.
But in general decidability in one context can't be expected to hold in another context. 
Of course this is already well-known.
Feynman formulated this difficult nature as follows for the two-slit experiment:
\begin{quote}
	``What we must say (to avoid making wrong predictions) is the following. If one looks at the holes or, more accurately, if one has a piece of apparatus which is capable of determining whether the electrons go through hole 1 or hole 2, then one \textit{can} say that it goes either through hole 1 or hole 2. \textit{But}, when one does \textit{not} try to tell which way the electron goes, when there is nothing in the experiment to disturb the electrons, then one may \textit{not} say that an electron goes either through hole 1 or hole 2. If one does say that, and starts to make any deductions from the statement, he will make errors in the analysis. This is the logical tightrope on which we must walk if we wish to describe nature successfully.'' \cite[p. 37-9]{Feynman63}\label{Feyn} 
\end{quote}
The aim is now to provide a logical framework in which this `logical tightrope' has a natural place.

\section{A simple example}

As may be obvious, in this text I take `measurement' as a primitive concept. 
It is my opinion that this is unproblematic as long as one focuses only on the epistemic part of a theory.
For quantum mechanics this seems justified since there is no consensus on the ontological interpretation.
A consequence of this approach is noted by Bell.
\begin{quote}
	``When one forgets the role of the apparatus, as the word `measurement' makes all too likely, one despairs of ordinary logic -- hence `quantum logic'. 
	When one remembers the role of the apparatus, ordinary logic is just fine.'' \cite[p. 34]{Bell90}
\end{quote}
And indeed the logic derived here will not be classical and may thus be conceived as a quantum logic.
However, I would like to go further than Bell and state that, if one finds an ontology for the apparatus in terms of the theory, classical logic seems a necessity; reality forces the law of excluded middle upon us.
In particular I am skeptic about the generalized notion of reality proposed in the topos theoretic approach in \cite{Isham11} or the `quantum numbers' approach in \cite{Corbett09}.
At the opposite end, it may be clear that classical logic is by no means sufficient for a clear ontological description \cite{Baltag11}.

The conceptually most difficult part in the derivation of a quantum logic is to let go of any ontological prejudices.
Therefore I will start with a simple example which will be generalized to classical mechanics.
It turns out that already in these cases the epistemic approach makes the logic significantly more complex.
In the case of classical mechanics the obtained logic can then easily be reduced to the standard logic by assuming the standard ontology.
That is, the interpretation in which observables correspond to certain elements of physical reality that have values at all times.
In the quantum case a similar logic will be developed but, as it is well known, due to the Kochen-Specker theorem the logic can not be reduced to a simpler one by introducing the standard ontology of classical mechanics.
The advantage of taking the detour with classical mechanics is that it shows already what an amount of work it takes to drop the assumption of realism and it shows the similarity between classical and quantum mechanics from the epistemic point of view.

It is the consensus that ``elementary'' propositions in scientific theories are of the form $A\in\Delta$ where $A$ is some observable and $\Delta$ is some subset of the set of all possible measurement outcomes for $A$ (usually taken to be $\mathbb{R}^n$ or some subset thereof). 
Oddly enough there is no consensus on what the proposition $A\in\Delta$ stands for exactly \cite{Isham95}. 
In the traditional viewpoint observables correspond to certain elements of a physical reality and, as such, have a definite value at all times. 
The proposition $A\in\Delta$ is then taken to be a proposition about the definite value of the element of physical reality corresponding to $A$; $A$ has a value in $\Delta$. 
It is an ontological proposition.

The simplest example is a theory describing only one observable $A$ which can assume the values 0 and 1. 
The standard logic for this theory consists of the sentences 
\begin{equation}\label{simplelog}
	\bot=A\in\varnothing~A\in\{0\}~A\in\{1\}~\top=A\in\{0,1\}.
\end{equation}
Equating $\top$ (triviality) with $A\in\{0,1\}$ is based on the assumption that $A$ has a definite value at all times. 
But from the instrumentalist point of view one cannot get to this conclusion; one is only certain of the fact that $A$ has a definite value when one checks that this is indeed the case.
Note that this is not the same as denying that $A$ has a definite value at all times, it is just acknowledging that instrumentally one cannot know this. 
It is based on the credo that ``unperformed experiments have no result''\footnote{This is the title of Peres' article \cite{Peres78} in which he advocates against the use of counterfactual reasoning. This opinion is also reflected in \cite{Peres84} and \cite{Peres02}.}.
So at least the sentences $\top$ and $A\in\{0,1\}$ should be treated as logically distinct.
It seems intuitive that a same argument can be held to take apart $\bot$ and $A\in\varnothing$.
But there is an asymmetry.
Still holding on to interpreting $A\in\Delta$ as ``$A$ has a value in $\Delta$'' the sentence $A\in\{0,1\}$ can be found true upon measurement of $A$. 
The sentence $A\in\varnothing$ on the other hand is a rather meaningless sentence in the instrumentalist approach.
In fact, if one thinks about it, the whole phrase ``$A$ has a value in'' is only well-defined in the context of a measurement.
So whenever the sentence $A\in\varnothing$ has a meaning, it is false.
But then a similar argument can be held for all sentences of the form $A\in\Delta$.

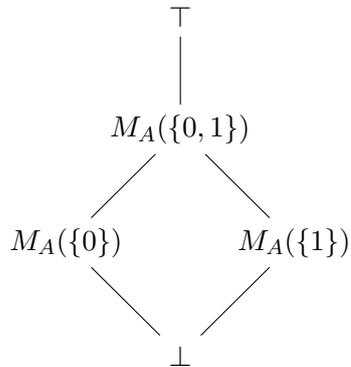
\begin{wrapfigure}{O}{0.5\textwidth}
\begin{center}
\begin{tikzpicture}[scale=1.5]
\path (0:0cm) node (v0) {$\bot$};
\path (1,1) node (v1) {$M_A(\{1\})$};
\path (-1,1) node (v2) {$M_A(\{0\})$};
\path (0,2) node (v3) {$M_A(\{0,1\})$};
\path (0,3) node (v4) {$\top$};
\draw (v0) -- (v1)
(v0) -- (v2)
(v1) -- (v3)
(v2) -- (v3)
(v3) -- (v4);
\end{tikzpicture}\caption{The intuitionistic logic for a simple theory with one observable.}\label{diagram2}
\end{center}
\end{wrapfigure}

The ambiguity that arises from the fact that sentences not always have a meaning is due to the ontological baggage that ``has a value'' brings along.
For the sake of clarity I therefore propose a new interpretation of $A\in\Delta$ and have it stand for ``I have measured $A$ and the result lay in $\Delta$''.
In this interpretation all the sentences in (\ref{simplelog}) have a definite meaning at all time. 
It also follows that $A\in\{0,1\}$ is not trivially true, whereas $A\in\varnothing$ may be considered trivially false.
For this last point I am adhering to the opposite of Peres' credo: ``performed experiments have a result''.
This idealization helps to keep things as simple as possible.
To emphasize this new interpretation I also introduce a new notation. 
The standard $A\in\Delta$ will now be denoted $M_A(\Delta)$. 
The logic resulting from these considerations is depicted in figure \ref{diagram2}.

Some readers may get the feeling they are being tricked into the use of intuitionistic logic so I would like to explain that in the present case the use of this kind of logic is not that strange at all.
Classical logic relies on the idea that there is a fact of the matter as to which every sentence is either true or false.
Intuitionistic logic focuses more on the epistemic view, interpreting `true' as `knowing it to be true'.
The truth of a negation is then read as `knowing it to be false'.
An example in mathematics is the sentence ``$\forall x\in\mathbb{R}:$ $x\geq0$ or $x\leq0$''.
Classically this is a true sentence, but intuitionists emphasize that real numbers can be constructed for which one simply does not know whether $x\geq0$  or $x\leq0$ is the case, and so they reject this sentence.
Platonists often argue that the sentence is true by virtue of the existence of the set of real numbers in a Platonic world. 
In physics intuitionistic logic then seems appropriate if one recognizes that one does not know what the ontology is of the system under investigation.

For the true classical logician these arguments are most likely not convincing and he/she may argue that negation is just wrongly defined here.
In that sense the negation of ``knowing it to be true'' should be ``not knowing it to be true'' (which is what one does in formal epistemic logic).
Much more propositions would have to be added to make the logic classical and most of them are rather dull.
For example, in the mathematical case one would now have propositions about not having a proof of $x\leq0$ for some $x\in\mathbb{R}$.
In physics one would have to add propositions about not performing measurements.
From this perspective the choice for intuitionistic logic in this text is based on simplicity.
So no claim for the necessity of intuitionistic logic is made here.
In particular, I disagree with the idea that logic may be empirical \cite{Putnam69}.

\section{Classical Mechanics}\label{cm}
In classical mechanics the situation immediately becomes more interesting and more complex simply because there is more than one observable.
The task is to identify elementary propositions of the form $M_A(\Delta)$ with mathematical objects in some set $L_{CM}$. 
Logically equivalent propositions may then be identified with the same object in $L_{CM}$. 
But also sentences composed of several of these elementary propositions should be in this set, and $L_{CM}$ should respect the reasoning structure for these sentences.
Consequently, $L_{CM}$ will be a lattice in which meets and joins can be interpreted as the logical connectives `and' and `or'.
The first thing to do is to establish what the total set of elementary propositions is i.e., to investigate which sentences $M_A(\Delta)$ can be formed.

In classical mechanics, an observable $A$ is identified with a function $f_A$ on a phase space $\Omega$ taking values in $V_A\subset\mathbb{R}$, the set of possible measurement outcomes for $A$. 
The laws about these observables are completely captured by the structure of these functions. 
A set of observables $\{A_1,\ldots,A_n\}$ obeys a law if there exists a function $f:V_{A_1}\times\cdots\times V_{A_n}\to\mathbb{R}$ such that
\begin{equation}
	f\left(f_{A_1}(\omega),\ldots,f_{A_n}(\omega)\right)=0,\quad\forall\omega\in\Omega.
\end{equation}
This law may then be written as $f(A_1,\ldots,A_n)=0$. 

It should be noted that I thus focus on the observables in the theory rather than on the phase space. 
In fact, I abandon the ontological assumption that the system at each moment actually finds itself in a certain state $\omega\in\Omega$.
So I also diverge from the standard method of developing a logic for classical mechanics, which is based on this ontological assumption.
The purpose is not to argue against this assumption but to show that, as it turns out, if one drops this assumption, classical mechanics and quantum mechanics become more alike. 

It seems inappropriate to assume that every function from $\Omega$ to some subset of $\mathbb{R}$ should also correspond to an observable. 
However, given some set of observables $\mathpzc{Obs}$, this set can be completed in a certain way, adding more observables. 
For instance, for every observable $A$ and every\footnote{It may disturb the reader that I allow every function and I don't restrict to some class of functions. 
It does indeed seem more appropriate to confine oneself to a certain class. However, the choice of which class is an entire different discussion and I don't have a clear opinion on this. 
The reader is free to choose any class (e.g. Borel, continuous, computable, etc.) for her/himself and it won't effect the present discussion.} 
$f:V_A\to\mathbb{R}$ the function $f\circ f_A$ can also be considered an observable corresponding to a measurement of $A$ and then applying $f$ to the outcome. 
More generally, for every sequence of observables $A_1,\ldots,A_n$ and every $f:V_{A_1}\times\cdots\times V_{A_n}\to\mathbb{R}$ one can consider the observable identified with a simultaneous measurement of $A_1,\ldots,A_n$ and then applying $f$ to the outcome string. 
A set of observables will be called complete if it contains all such additional observables (i.e. it is closed under the application of functions).
I will assume from now on that $\mathpzc{Obs}$ is completed in this sense.

The set of elementary propositions for classical mechanics is now given by
\begin{equation}
	EP_{CM}:=\{M_A(\Delta)\:;\:A\in\mathpzc{Obs},\Delta\subset V_A\}.
\end{equation}
In the end, this set should in some way be related to the logic $L_{CM}$ of the theory. 
To figure out what this logic should be it is best to investigate how this logic should behave when restricting to elementary propositions.
In other words, it is to be investigated what an appropriate preorder would be for the set $EP_{CM}$ both interpretation wise and mathematically (for indeed, the elementary propositions have now been identified with mathematical objects).

A consequence of the functional relationships between observables is that the set of possible measurement outcomes for an observable isn't of much importance, but rather the partition of $\Omega$ that is generated by it, i.e. the set
\begin{equation}
	P_A:=\{f_A^{-1}(x)\:;\:x\in V_A\}.
\end{equation}
In particular, if two observables generate the same partition, then they are completely equivalent. 
That is, a measurement of the one allows one to completely determine the outcome of the measurement of the other and vice versa. 
Note that this is the case for two observables $A_1$ and $A_2$ iff there exists an invertible function $f:V_{A_1}\to V_{A_2}$ such that $f\circ f_{A_1}=f_{A_2}$ and $f^{-1}\circ f_{A_2}=f_{A_1}$. 
Consequently, a sufficient condition for two propositions $M_{A_1}(\Delta_1)$ and $M_{A_2}(\Delta_2)$ to be logically equivalent is $P_{A_1}=P_{A_2}$ and $f_{A_1}^{-1}(\Delta_1)=f_{A_2}^{-1}(\Delta_2)$.
This logical equivalence is a first indication of what the pre-order on $EP_{CM}$ should look like.
In the end two propositions $M_{A_1}(\Delta_1)$ and $M_{A_2}(\Delta_2)$ are then logically equivalent iff
\begin{equation}
	M_{A_1}(\Delta_1)\leq M_{A_2}(\Delta_2)~\text{and}~M_{A_2}(\Delta_2)\leq M_{A_1}(\Delta_1).
\end{equation}
In the following I define $\leq$ by investigating sufficient and necessary conditions for concluding $M_{A_1}(\Delta_1)\leq M_{A_2}(\Delta_2)$ for any pair of propositions.
Obviously, $\leq$ must therefore be interpreted as ``logically implies''. 

The set of all partitions generated by $\mathpzc{Obs}$ 
\begin{equation}
	L_{\mathpzc{Obs}}:=\{P_A\:;\:A\in\mathpzc{Obs}\},
\end{equation}
will play an important role. It is turned into a lattice by the following definitions:
\begin{subequations}
\begin{equation}
	P_{A_1}\leq P_{A_2}\quad\Longleftrightarrow\quad \forall \Delta_1\in P_{A_1},\exists \Delta_2\in P_{A_2}:\Delta_1\subset\Delta_2,
\end{equation}
\begin{equation}\label{partitionmeet}
	P_{A_1}\wedge P_{A_2}:=\{\Delta_1\cap\Delta_2\:;\:\Delta_1\in P_{A_1},\Delta_2\in P_{A_2}\},
\end{equation}
\begin{equation}\label{partitionjoin}
	P_{A_1}\vee P_{A_2}:=\bigwedge\{P_A\:;\:A\in \mathpzc{Obs}, P_{A_1}\leq P_A, P_{A_2}\leq P_A\}.
\end{equation}
\end{subequations}
Note that $P_{A_1}\leq P_{A_2}$ iff there is a surjective function $f:V_{A_1}\to V_{A_2}$. 
It should be checked that the operations (\ref{partitionmeet}) and (\ref{partitionjoin}) again correspond to elements of $L_{\mathpzc{Obs}}$. 
This indeed follows from the criterion that $\mathpzc{Obs}$ was completed. 
For example, $P_{A_1}\wedge P_{A_2}$ corresponds to the partition of any observable $A$ for which there is an invertible function $f:V_{A_1}\times V_{A_2}\to V_A$. 
Arbitrary meets can also be defined and consequently, there is also a bottom element $P_{\mathpzc{Obs}}$ corresponding to the measurement of all observables. 
Note that if $\mathpzc{Obs}$ would include all functions, this bottom element becomes the partition of $\Omega$ in singleton sets $\{\{\omega\}\:;\:\omega\in\Omega\}$.
The top element of $L_{\mathpzc{Obs}}$ is given by the partition $\{\Omega\}$.

There is a direct physical significance to this lattice.
The existence of a surjective function $f:V_{A_1}\to V_{A_2}$ implies that there is a law associating with each outcome of the measurement of $A_1$ a distinct outcome for the measurement of $A_2$.
Certainly, within the theory, these laws are true indefinitely and therefore performing a measurement of $A_1$ and then applying the function $f$ to the outcome counts as a proper measurement of $A_2$.
By the same reasoning, suppose $P_{A_3}=P_{A_1}\wedge P_{A_2}$. 
Then a measurement of $A_3$ counts as a measurement of both $A_1$ and $A_2$.
And in the case that $P_{A_3}=P_{A_1}\vee P_{A_2}$, then both a measurement of $A_1$ counts as a measurement of $A_3$ and a measurement of $A_2$ counts as one for $A_3$. 

Besides the structure of $L_{\mathpzc{Obs}}$, the lattices of subsets of the outcome sets also play an important role. 
If one keeps the observable $A$ fixed, it is clear that one should have that if $\Delta\subset\Delta'\subset V_A$ then
\begin{equation}\label{subss}
	M_A(\Delta)\leq M_A(\Delta').
\end{equation}
So this establishes what the preorder on $EP_{CM}$ should be when restricting to a single observable.
To investigate the relation between elementary propositions concerning distinct observables consider first the situation where two observables $A_1$ and $A_2$ satisfy the relation $P_{A_1}\leq P_{A_2}$.
From earlier considerations I argued that this together with the assumption $f_{A_1}^{-1}(\Delta_1)=f_{A_2}^{-1}(\Delta_2)$ is sufficient to conclude that $M_{A_1}(\Delta_1)$ implies $M_{A_2}(\Delta_2)$.
Combining this with (\ref{subss}) results in the conclusion
\begin{equation}\label{vglx}
	\left(P_{A_1}\leq P_{A_2} \text{ and } f_{A_1}^{-1}(\Delta_1)\subset f_{A_2}^{-1}(\Delta_2)\right)
	\Longrightarrow
	M_{A_1}(\Delta_1)\leq M_{A_2}(\Delta_2).
\end{equation} 

In this situation the existence of a particular functional relation (law) between $A_1$ and $A_2$ is used to argue that a measurement of $A_1$ counts as a measurement of $A_2$. 
But does conversely a measurement of $A_2$ furnish any information about $A_1$?
In common texts on logics for classical mechanics the answer is definitely yes. 
The proposition $M_{A_2}(\Delta_2)$ would then imply $M_{A_1}(f^{-1}(\Delta_2))$. 
This conclusion is based on the reasoning that a measurement of $A_2$ reveals information about the state $\omega\in\Omega$ in which the system finds itself, 
and in all the possible states in which $M_{A_2}(\Delta_2)$ is true, the proposition $M_{A_1}(f^{-1}(\Delta_2))$ is also true. 
Here one interprets $M_A(\Delta)$ as a proposition about a property of the system.
However, this use of terminology is unsuiting if one abandons the realist interpretation of the state of a system. 
In fact, it may be clear by now that if one focuses on the concept of measurement the proposition $M_{A_1}(V_{A_1})$ is a stronger one than $M_{A_2}(V_{A_2})$ if $P_{A_1}\leq P_{A_2}$.  

So in this light it seems natural to just define $\leq$ by replacing `$\Longrightarrow$' by `$\Longleftrightarrow$' in (\ref{vglx}). 
I would say this is almost correct. 
There are actually situations in which the order structure of $L_{\mathpzc{Obs}}$ doesn't matter, and that is when the elementary proposition itself is a contradiction: $M_A(\varnothing)$. 
Remember that this is a consequence of the idealization that performed experiments have a result. 
With this idealization the preorder on $EM_{CM}$ becomes
\begin{equation}
\begin{gathered}
	M_{A_1}(\Delta_1)\leq M_{A_2}(\Delta_2)\\ 
	\Longleftrightarrow\\ 
	\left(P_{A_1}\leq P_{A_2} \text{ and }f^{-1}_{A_1}(\Delta_1)\subset f^{-1}_{A_2}(\Delta_2)\right)\text{ or }\Delta_1=\varnothing.
\end{gathered}
\end{equation}
This preorder leads to a non-trivial equivalence relation, and the set of equivalence classes can now be characterized by the set
\begin{equation}
	S_{CM}:=\{(P,\Delta)\:;\:P\in L_{Obs}, \varnothing\subsetneq\Delta\subset\Omega, P\leq\{\Delta,\Delta^c\}\}\cup\{\bot\},
\end{equation}
where $\bot$ corresponds with the equivalence class $\{M_A(\varnothing)\:;\:A\in\mathpzc{Obs}\}$.
The inherited partial order on $S_{CM}$ takes the form
\begin{equation}
	(P_1,\Delta_1)\leq (P_2,\Delta_2)\Longleftrightarrow P_1\leq P_2\text{ and }\Delta_1\subset\Delta_2,
\end{equation}
and of course $\bot\leq (P,\Delta)$ for all $(P,\Delta)$.
It should be noted that already $S_{CM}$ has a much richer structure than the standard logic for classical mechanics; there is no way to associate every $(P,\Delta)$ with a subset of the state space in a consistent way without `forgetting about $P$'.

Now in the end the set $S_{CM}$ of equivalence classes of elementary propositions will be considered as a subset of the logic $L_{CM}$.
But $L_{CM}$ also includes disjunctions and conjunctions of sentences in $S_{CM}$ and may therefore be larger than $S_{CM}$. 
It must be investigated if this is indeed the case, and I will start with conjunctions.

Consider two observables $A_1,A_2\in\mathpzc{Obs}$. 
Then for any pair of sets $\Delta_1,\Delta_2$ the proposition $M_{A_1}(\Delta_1)\wedge M_{A_2}(\Delta_2)$ is read as ``I have measured $A_1$ and the result lay in $\Delta_1$ and I have measured $A_2$ and the result lay in $\Delta_2$''. 
If the use of the second `and' in this sentence is in any sense similar to the uses of the other `and'-s (a natural requirement), this may also be read as ``I have measured $A_1$ and $A_2$ and the result of the first lay in $\Delta_1$ and the result of the second in $\Delta_2$''.
By assumption there is an observable $A_3$ whose measurement counts as a measurement of both $A_1$ and $A_2$ and that satisfies $P_{A_3}=P_{A_1}\wedge P_{A_2}$.
The above sentence is then equivalent to the sentence ``I have measured $A_3$ and $f_1$ applied to the result lay in $\Delta_1$ and $f_2$ applied to the result lay in $\Delta_2$'', where $A_1=f_1(A_3)$ and $A_2=f_2(A_3)$.
But this is just the same as saying that ``I have measured $A_3$ and the result lay in $f_1^{-1}(\Delta_1)\cap f_2^{-1}(\Delta_2)$'', which again corresponds to an elementary proposition.
In conclusion, in $L_{CM}$ it should at least hold that
\begin{equation}
	(P_1,\Delta_1)\wedge(P_2,\Delta_2)=\begin{cases} (P_1\wedge P_2,\Delta_1\cap\Delta_2),& \Delta_1\cap\Delta_2\neq\varnothing\\ \bot, &\text{else.}\end{cases}
\end{equation}

So conjunctions are pretty much what one would expect.
Disjunctions on the other hand are more difficult because they cannot be imbedded within the partial ordered set of elementary propositions. 
At least, not without running into interpretational difficulties. 
There is of course the option to define it as the least upper bound given the definition of conjunctions:
\begin{equation}
\begin{split}
	(P_{A_1},\Delta_1)\vee(P_{A_2},\Delta_2)
	&:=
	\bigwedge\left\{(P_A,\Delta)\in S_{CM}\:;\: \substack{(P_{A_1},\Delta_1)\leq (P_{A},\Delta),\\ (P_{A_2},\Delta_2)\leq (P_{A},\Delta)}\right\}\\
	&=
	(P_{A_1}\vee P_{A_2},\Delta_1\cup\Delta_2)
\end{split}
\end{equation}
but it seems inappropriate to identify a measurement of $A_1$ or $A_2$ with a measurement of neither.
The present interpretation demands something stronger.
But it is also impossible to identify the disjunction with a joint measurement of the two observables. 
This is too strong a demand and in conflict with the partial order on $S_{CM}$ (i.e. logically inconsistent).
Rather, the aim is to broaden the lattice such that
\begin{equation}
	\left(P_{A_1}\wedge P_{A_2},\Delta_1\cup \Delta_2\right)
	< 
	(P_{A_1},\Delta_1)\vee(P_{A_2},\Delta_2)
	<
	\left(P_{A_1}\vee P_{A_2},\Delta_1\cup \Delta_2\right).
\end{equation}
That is, the disjunction of two elementary propositions is simply no longer an elementary proposition.

The set of propositions increases immensely with this step, for it now also includes all propositions of the form
\begin{equation}\label{disjx}
	\bigvee_{A\in O}(P_A,\Delta_A),\quad O\subset\mathpzc{Obs}, (P_A,\Delta_A)\in S_{CM}.
\end{equation}
Fortunately, conjunctions of such propositions can then again be defined by postulating that the lattice should be distributive.
The equality
\begin{equation}
	 \bigvee_{A_1\in O_1}(P_{A_1},\Delta_{A_1})\wedge \bigvee_{A_2\in O_2}(P_{A_2},\Delta_{A_2})
	 =
	 \bigvee_{\substack{A_1\in O_1\\ A_2\in O_2}}(P_{A_1},\Delta_{A_1})\wedge(P_{A_2},\Delta_{A_2})
\end{equation}
will then simply be taken as the definition of the left-hand side.
I will now investigate how the set of propositions can be fully specified. 

Note that every proposition of the form (\ref{disjx}) can be written as a disjunction over all observables simply by introducing a somewhat sloppy notation and taking disjunctions with $(P_A,\varnothing)=\bot$ for all $A\notin O$. 
Of course this is a cumbersome way of expressing a proposition but it has the advantage that the same notation can be used for every propositions, elementary or not.
More specifically, every proposition can be written as a function $S:L_{\mathpzc{Obs}}\to \mathcal{P}(\Omega)$ (the power set of $\Omega$) with the restriction that for all $A$ $(P_A,S(P_A))\in S_{CM}$. 
This function can then be identified with a proposition in the following way:\footnote{The introduction of these functions is purely for mathematical convenience. 
One may also formally introduce the set of objects of the form of (\ref{disjx}) and then introduce disjunctions that are consistent with this notation.}
\begin{equation}\label{intp}
	S\simeq \bigvee_{P\in L_{\mathpzc{Obs}}}(P,S(P)).
\end{equation}
Clearly the set
\begin{equation}
	F:=\{S:L_{\mathpzc{Obs}}\to \mathcal{P}(\Omega)\:;\:(P,S(P))\in S_{CM}\text{ or }S(P)=\varnothing\}
\end{equation}
with its interpretation (\ref{intp}) is rich enough to incorporate the desired disjunctions, but it also needs to be narrowed down for it contains many sentences that are logically equivalent.

With a single proposition $(P,\Delta)$ several functions can be identified. 
Indeed, any function will do as long as $S(P)=\Delta$ and $(P',S(P'))\leq (P,\Delta)$ for all $P'\in L_{\mathpzc{Obs}}$.
Note that this necessitates that $S(P')=\varnothing$ whenever neither $P\leq P'$ nor $P'\leq P$.
That is, as long as $S$ represents the disjunction of $(P,\Delta)$ with sentences that are all stronger than $(P,\Delta)$, then the meaning of $S$ is equivalent to $(P,\Delta)$.
A special procedure for constructing such a function is taking for each $S(P')$ the highest value possible in accordance with the interpretation of $(P,\Delta)$. 
That is, $(P',S(P'))$ is the weakest sentence that is still stronger than $(P,\Delta)$.
The function obtained in this way will be denoted $S_{(P,\Delta)}$ and it satisfies
\begin{equation}
	S_{(P,\Delta)}(P'):=\begin{cases}
	\Delta,& P'\leq P\\
	\varnothing, \text{else}.
	\end{cases}
\end{equation}
In the same line, the function $S_{\bot}$ associated with $\bot\in S_{CM}$ is the function that assigns $\varnothing$ to every partition.
With this definition an injection $i:S_{CM}\to F$ has been defined by
\begin{equation}
	i:(P,\Delta)\mapsto S_{(P,\Delta)}.
\end{equation}

The disjunctions in the form of (\ref{disjx}) can now officially be defined in terms of these functions:
\begin{equation}\label{disjcm}
	\left(\bigvee_{A\in O}S_{(P_A,\Delta_A)}\right)(P):=\bigcup_{A\in O}\left(S_{(P_A,\Delta_A)}(P)\right).
\end{equation}
The set of all functions of this form forms a subset of $F$ and it is given by
\begin{equation}
	L_{CM}:=\{S:L_{\mathpzc{Obs}}\to \mathcal{P}(\Omega)\:;\: S\in F\text{ and } S(P_1)\subset S(P_2) \text{ whenever }P_1\geq P_2\}.
\end{equation} 
Indeed, it is straight forward to check that the construction (\ref{disjcm}) always leads to an element of $L_{CM}$. 
Conversely, every element of $L_{CM}$ is of this form since every $S\in L_{CM}$ satisfies
\begin{equation}
	S=\bigvee_{P\in L_{\mathpzc{Obs}}}S_{(P,S(P))}.
\end{equation}
That is, every element of $L_{CM}$ is associated with a disjunction of (equivalence classes of) elementary propositions.
The set $L_{CM}$ is turned into a complete distributive lattice by the following definitions:
\begin{subequations}
\begin{equation}
	S_1\leq S_2 \Longleftrightarrow S_1(P)\subset S_2(P) \forall P\in L_{\mathpzc{Obs}},
\end{equation}
\begin{equation}
	\left(\bigvee_{i\in I}S_i\right)(P):=\bigcup_{i\in I}S_i(P),
\end{equation}
\begin{equation}
	\left(\bigwedge_{i\in I}S_i\right)(P):=\bigcap_{i\in I}S_i(P).
\end{equation}
\end{subequations}
Infinite distributivity follows from the infinite distributivity of the Boolean lattice of subsets of $\Omega$. 
The bottom element is given by $S_\bot$ and the top element is given by $S_\top$ which assigns to each partition the set $\Omega$. 
Finally, $L_{CM}$ is turned into a Heyting algebra by defining the relative pseudo complement
\begin{equation}
	S_1\to S_2:=\bigvee\{S\in L_{CM}\:;\:S\wedge S_1\leq S_2\}.
\end{equation}

It is easy to check that the embedding function $i$ respects the partial order and preserves arbitrary meets:
\begin{equation}
\begin{gathered}
	(P_1,\Delta_1)\leq (P_2,\Delta_2) \Longleftrightarrow S_{(P_1,\Delta_1)}\leq S_{(P_2,\Delta_2)},\\
	S_{(P_1,\Delta_1)\wedge (P_2,\Delta_2)}= 	S_{(P_1,\Delta_1)}\wedge S_{(P_2,\Delta_2)}.
\end{gathered}
\end{equation}
And this turns $L_{CM}$ into a proper extension of $S_{CM}$.

Strictly speaking, the top element $S_\top$ of $L_{CM}$ corresponds with the measurement of any observable taking only one possible outcome e.g. $f_\top:\Omega\to\{1\}$. 
A measurement of such an observable may be considered to be not a measurement at all, for it doesn't require any physical act to obtain the value of this observable; it is already given by the laws of the theory.  
In the case where $\mathpzc{Obs}$ is generated by a single observable $A$ with $V_A=\{0,1\}$, this lattice reduces to the one in figure \ref{diagram2}. 
To see this note that the lattice of partitions is given by
\begin{equation}
	L_{\mathpzc{Obs}}=\{P_A,P_0\}, 
\end{equation}
with $P_A$ the partition of $\Omega$ generated by $A$ given by
\begin{equation}
	P_A=\{\Delta_0=f_A^{-1}(\{0\}),\Delta_1=f_A^{-1}(\{1\})\}
\end{equation}
and $P_0$ the trivial partition $\{\Omega\}$.
Notating the elements of $L_{CM}$ as pairs $(S(P_0),S(P_A))$, one finds that 
\begin{equation}
	L_{CM}=\{\bot=(\varnothing,\varnothing),(\varnothing,\Delta_0),(\varnothing,\Delta_1),(\varnothing,\Omega),(\Omega,\Omega)=\top\},
\end{equation} 
which is precisely the lattice of figure \ref{diagram2}.

It takes some time to appreciate the complexity of $L_{CM}$.
Obviously, assuming the standard ontology for classical mechanics is much more convenient in every day life.
But now imagine (if you can) a person unaware of any realist interpretation of classical mechanics.
Then his/her logic for reasoning in classical mechanics is likely to resemble $L_{CM}$.
In fact, this is the situation we find ourselves in with respect to quantum mechanics.
With the difference of not having a formal logic $L_{QM}$.
So what is again the advantage of such a logic?
First of all, it makes clearer what is speakable in quantum mechanics independent of any realist interpretation.
And secondly, realist interpretations can be compared with respect to their simplification of $L_{QM}$.
It makes clear the explaining role of the interpretation.
For example, consider meeting the aforementioned person.
At first you are confused about his/her form of reasoning.
But then you recognize the form of $L_{CM}$ and you see that the reasoning is correct, but just very cumbersome.
And then you can explain that propositions like $S_{(P_1,\Delta)}$ and $S_{(P_2,\Delta)}$ actually are equivalent.
And then this person can reflect about whether or not this interpretation is satisfactory.
Are the philosophical consequences of the interpretation satisfactory or not?
I will return to this discussion in section \ref{discussie}.  

\section{Quantum Mechanics}\label{IQL}
For the definition of $L_{QM}$ I will follow an approach similar to the one for classical mechanics.
But of course, since it is a different theory, along the line differences in the logic will emerge.
In quantum mechanics a system is associated with a C*-algebra $\mathcal{C}$ and each observable $A$ is associated with a self-adjoint operator $\hat{A}$ within this algebra. 
For convenience I will assume that every self-adjoint operator is associated with an observable. 
In section \ref{discussie} it will follow that this assumption can be relaxed. 
Also, for mathematical convenience, I will only consider finite dimensional C*-algebras. 
More general cases can be studied, but the necessary mathematical care involved would soon blur the discussion. 
The set of possible measurement outcomes for an observable $A$ is given by the spectrum $\sigma(\hat{A})$ of the operator $\hat{A}$. 
The set of elementary propositions for quantum mechanics is thus given by
\begin{equation}
	EP_{QM}:=\{M_A(\Delta)\:;\:\hat{A}=\hat{A}^*, \Delta\subset\sigma(\hat{A})\}.
\end{equation}

As with the classical case, both elements of the pair that constitute an elementary proposition bring along a mathematical structure. 
This then imposes a preorder on $EP_{QM}$. 
The correct preorder depends on the interpretation of the elementary propositions.

A measurement of any observable $A$ may also count as a measurement of the observable obtained by applying a function $f$ to the outcome of the measurement of $A$. 
This observable should again be associated with a self-adjoint operator. 
Luckily such an operator can be found easily by applying $f$ to the operator associated with $A$.\footnote{Formally, the function of an operator isn't well-defined for arbitrary $f$, but it is for a huge class off functions such as Borel functions.} 
For every observable $A$ there is a unique Abelian sub-algebra $\mathcal{A}\subset\mathcal{C}$ in which all self-adjoint elements are functions of the operator associated with $A$, and every function of $A$ is an element of $\mathcal{A}$. 
Thus a measurement of $A$ can be associated with a measurement of all the observables whose corresponding operator lies in $\mathcal{A}$. 
Every Abelian sub-algebra of $\mathcal{C}$ is of this form. 
Two observables are considered equivalent if they generate the same algebra. 
The set of all Abelian sub-algebras is denoted $\mathfrak{A}(\mathcal{C})$.

Like with the partitions of $\Omega$ in the classical case, the Abelian sub-algebras of $\mathcal{C}$ form a partial ordered set by taking set inclusion as partial order. 
It should be noted that, like in the earlier example of incompatible observables, this poset is not a lattice. 
A meet is still defined by taking the intersection, but in general for a pair of Abelian algebras $\mathcal{A}_1$ and $\mathcal{A}_2$ there is no Abelian algebra containing both. 
In fact, this is only the case if the algebras commute: $[\mathcal{A}_1,\mathcal{A}_2]=0$. 
That is, every element of $\mathcal{A}_1$ commutes with every element of $\mathcal{A}_2$.
It is important to note that the partial order is reversed in comparison with the classical case. 
That is, for two observables $A_1$ and $A_2$ with $A_2=f(A_1)$ one has $P_{A_1}\leq P_{A_2}$ in the classical case and $\mathcal{A}_2\subseteq\mathcal{A}_1$ in the quantum case.

In quantum mechanics the set of possible measurement outcomes $V_A$ coincides with the spectrum of the operator associated with $A$. 
Furthermore, every subset $\Delta$ of $V_A$ can be associated with a projection operator $\mu_{\hat{A}}(\Delta)\in\mathcal{A}$, where $\mu_{\hat{A}}$ is the spectral measure for $\hat{A}$. 
These operators correspond with the observable associated with applying the function that assigns the value $1$ to elements of $\Delta$ and the value $0$ to the complement of $\Delta$ in $V_A$ to the outcome of a measurement of $A$. 
All projection operators in $\mathcal{A}$ are of this form and their set is denoted $\mathcal{P}(\mathcal{A})$. 
The analogy with the classical case thus far is summarized as follows:
\begin{eqnarray}
f_A:\Omega\to V_A & \leftrightarrow & \hat{A}\in\mathcal{C},\quad \hat{A}=\hat{A}^*\\
P_A\in L_{\mathpzc{Obs}} & \leftrightarrow & \mathcal{A}\in\mathfrak{A}(\mathcal{C})\\
\{\Delta\subset\Omega\:;\:P_A\leq\{\Delta,\Delta^c\}\} & \leftrightarrow & \mathcal{P}(\mathcal{A})
\end{eqnarray}

In orthodox quantum logic the partial order structure on the set of observables is completely ignored. 
There, two elementary propositions $M_{A_1}(\Delta_1)$ and $M_{A_2}(\Delta_2)$ are considered to be equal iff $\mu_{\hat{A}_1}(\Delta_1)=\mu_{\hat{A}_2}(\Delta_2)$. 
It thus assumes the partial order $M_{A_1}(\Delta_1)\leq M_{A_2}(\Delta_2)$ iff $\mu_{\hat{A}_1}(\Delta_1)\leq\mu_{\hat{A}_2}(\Delta_2)$. 
In other words, the proposition $M_{A_1}(\Delta_1)$ is being equated to the proposition $M_{\mu_{\hat{A}_1}(\Delta_1)}(\{1\})$. 
The underlying thought is most likely that both propositions reveal the same information about the actual state of the system. 
But although that interpretation may be consistent for classical mechanics, in quantum mechanics it is controversial. 
In fact, there aren't many people left who think this line of ontological reasoning leads to satisfactory results.
That is the tale of quantum logic \cite{Maudlin05}.

One of my main objections is that orthodox quantum logic doesn't take into account the notion of incompatible observables. 
Indeed, it may well be the case that $\mu_{\hat{A}_1}(\Delta_1)=\mu_{\hat{A}_2}(\Delta_2)$ even when $A_1$ and $A_2$ do not commute. 
It seems hardly appropriate that two propositions can be considered equivalent even if they are about mutually exclusive measurements. 
Adding the requirement $[\hat{A}_1,\hat{A}_2]=0$ as being necessary for equivalence isn't enough, for it conflicts with associativity: 
\begin{equation}
	M_{A_1}(\Delta_1)\sim M_{\mu_{\hat{A}_1}(\Delta_1)}(\{1\})\sim M_{\mu_{\hat{A}_2}(\Delta_2)}(\{1\})\sim M_{A_2}(\Delta_2),
\end{equation}
while $M_{A_1}(\Delta_1)\nsim M_{A_2}(\Delta_2)$.

The solution is to take the notion of measurement seriously, leading to the definition
\begin{equation}
	M_{A_1}(\Delta_1)\sim M_{A_2}(\Delta_2)\quad\Longleftrightarrow\quad \mathcal{A}_1=\mathcal{A}_2\text{ and }\mu_{\hat{A}_1}(\Delta_1)=\mu_{\hat{A}_2}(\Delta_2),
\end{equation}
where it is assumed that $\Delta_1$ and $\Delta_2$ are not empty. 
Propositions of the form $M_A(\varnothing)$ are again identified with contradiction.
The set of (equivalence classes of) elementary propositions is characterized by the set
\begin{equation}
	S_{QM}:=\{(\mathcal{A},P)\:;\:\mathcal{A}\in\mathfrak{A}(\mathcal{C}), P\in\mathcal{P}(\mathcal{A}), P\neq0\}\cup\{\bot\},
\end{equation}
with injection $M_A(\Delta)\mapsto(\mathcal{A},\mu_{\hat{A}}(\Delta))$. 
The partial order is defined analogously to the classical case:
\begin{equation}
	(\mathcal{A}_1,P_1)\leq (\mathcal{A}_2,P_2)\Longleftrightarrow (\mathcal{A}_2\subset\mathcal{A}_1\text{ and }P_1\leq P_2)\text{ or } P_1=0,
\end{equation}
where all elements of the form $(\mathcal{A},0)$ are considered equivalent and equal to $\bot$.

Like in the classical case, it should be investigated how disjunctions and conjunctions of elementary propositions behave.
Now again it is on the wish list to associate the conjunction $M_{A_1}(\Delta_1)\wedge M_{A_2}(\Delta_2)$ with a proposition concerning the simultaneous measurement of $A_1$ and $A_2$.
In quantum mechanics this is troublesome if $[\hat{A_1},\hat{A_2}]\neq0$.
In fact, I propose that in this case $M_{A_1}(\Delta_1)\wedge M_{A_2}(\Delta_2)$ simply expresses a contradiction.
On the other hand, if $\hat{A_1}$ and $\hat{A_2}$ do commute, the conjunction has a clear meaning as an elementary proposition:
\begin{equation}\label{QMconj}
	(\mathcal{A}_1,P_1)\wedge (\mathcal{A}_2,P_2)=\begin{cases}
	(\mathcal{A}_1\vee\mathcal{A}_2,P_1\wedge P_2), & [\mathcal{A}_1,\mathcal{A}_2]=0\\
	\bot, & \text{else},
	\end{cases}
\end{equation}
where $\mathcal{A}_1\vee\mathcal{A}_2$ is the smallest Abelian sub-algebra that has both $\mathcal{A}_1$ and $\mathcal{A}_2$ as a subset.
However, for disjunctions one faces the same difficulties as for the classical case and the set of propositions has to be expanded.
For this I take a similar approach as for the classical case.

A measurement of $A_1$ implies a measurement of $A_2$ iff there is a function $f$ such that $A_2=f(A_1)$. 
This is the case iff $\mathcal{A}_2\subset\mathcal{A}_1$. 
A proposition $(\mathcal{A},P)$ is thus equivalent with the disjunction of all propositions $(\mathcal{A}',P')$ with $\mathcal{A}\subset\mathcal{A}'$ and $P'\leq P$. 
This is then again equivalent to the disjunction of all propositions $(\mathcal{A}',P)$ with $\mathcal{A}\subset\mathcal{A}'$. 
The proposition $(\mathcal{A},P)$ can thus again be identified with a function $S_{(\mathcal{A},P)}:\mathfrak{A}(\mathcal{C})\to \mathcal{P}(\mathcal{C})$ with
\begin{equation}
	S_{(\mathcal{A},P)}(\mathcal{A}')=\begin{cases} P, &\text{if }\mathcal{A}\subset\mathcal{A}',\\ 0,&\text{else.}\end{cases}
\end{equation}

$S_{(\mathcal{A},P)}$ is interpreted as the disjunction of all propositions $(\mathcal{A}',S(\mathcal{A}'))$ where $\mathcal{A}'$ runs over all Abelian sub-algebras of $\mathcal{C}$. 
By assuming associativity\footnote{Or rather, distributivity between the disjunctions within the interpretation of each $S_i$ and the disjunction between $S_1$ and $S_2$.}, disjunctions of such propositions $S_i:\mathfrak{A}(\mathcal{C})\to \mathcal{P}(\mathcal{C})$ can be formed by taking joins on the level of the projection operators:
\begin{equation}\label{QMdisj}
	(S_1\vee S_2)(\mathcal{A}):=S_1(\mathcal{A})\vee S_2(\mathcal{A}),
\end{equation}
which again defines a function on $\mathfrak{A}(\mathcal{C})$. 
The set of functions that is formed by taking consecutive disjunctions in this manner is given by
\begin{equation}
	L_{QM}:=\left\{S:\mathfrak{A}(\mathcal{C})\to \mathcal{P}(\mathcal{C})\:;\:
	\substack{S(\mathcal{A})\in \mathcal{P}(\mathcal{A})\text{ and}\\
	S(\mathcal{A}_1)\leq S(\mathcal{A}_2)\text{ whenever } \mathcal{A}_1\subset\mathcal{A}_2}\right\}.
\end{equation} 
Indeed, for every index function $I$ one has that $\bigvee_{i\in I}S_{(\mathcal{A}_i,P_i)}$ is an element of $L_{QM}$. 
On the other hand, every element of $L_{QM}$ satisfies
\begin{equation}
	S=\bigvee_{\mathcal{A}\in\mathfrak{A}(\mathcal{C})}S_{(\mathcal{A},S(\mathcal{A}))}.
\end{equation}	 

Obviously, (\ref{QMdisj}) cannot be a join until a partial order has been defined on $L_{QM}$. 
This partial order is inherited from the partial order on $S_{QM}$. 
To say that $S_1\leq S_2$ is to say that for every $\mathcal{A}_1\in\mathfrak{A}(\mathcal{C})$ there exists an $\mathcal{A}_2\in\mathfrak{A}(\mathcal{C})$ such that $(\mathcal{A}_1,S_1(\mathcal{A}_1))\leq(\mathcal{A}_2,S_2(\mathcal{A}_2))$. 
This requires that $\mathcal{A}_1\supset\mathcal{A}_2$ and $S_1(\mathcal{A}_1)\leq S_2(\mathcal{A}_2)$. 
Because $S_2\in L_{QM}$ this implies $S_1(\mathcal{A}_1)\leq S_2(\mathcal{A}_1)$. 
The partial order then obtained is
\begin{equation}
	S_1\leq S_2\quad\Longleftrightarrow\quad S_1(\mathcal{A})\leq S_2(\mathcal{A})~ \forall \mathcal{A}\in\mathfrak{A}(\mathcal{C}).
\end{equation}
It is easy to check that (\ref{QMdisj}) is indeed the join with respect to this order.

The meet on $L_{QM}$ is now given by 
\begin{equation}\label{QMconj2}
	(S_1\wedge S_2)(\mathcal{A}):=S_1(\mathcal{A})\wedge S_2(\mathcal{A}).
\end{equation}
It is consistent with the conjunction for the elementary propositions given by (\ref{QMconj2}):
\begin{equation}
	S_{(\mathcal{A}_1,P_1)\wedge(\mathcal{A}_2,P_2)}=S_{(\mathcal{A}_1,P_1)}\wedge S_{(\mathcal{A}_2,P_2)}.
\end{equation}

The lattice $L_{QM}$ is furthermore a complete distributive lattice by virtue of the lattice structure of the projection operators:
\begin{equation}
\begin{split}
	\left(\bigvee_{i\in I}(\bigwedge_{j\in J} S_{i,j})\right)(\mathcal{A})
	&=
	\bigvee_{i\in I}\left(\left(\bigwedge_{j\in J}S_{i,j}\right)(\mathcal{A})\right)
	=
	\bigvee_{i\in I}\left(\left(\bigwedge_{j\in J}S_{i,j}(\mathcal{A})\right)\right)\\
	&=
	\bigwedge_{j\in J}\left(\left(\bigvee_{i\in I}S_{i,j}(\mathcal{A})\right)\right)
	=
	\bigwedge_{j\in J}\left(\left(\bigvee_{i\in I}S_{i,j}\right)(\mathcal{A})\right)\\
	&=
	\left(\bigwedge_{j\in J}(\bigvee_{i\in I} S_{i,j})\right)(\mathcal{A}).
\end{split}
\end{equation}
Consequently, $L_{QM}$ is turned into a Heyting algebra by introducing the relative pseudo-complement
\begin{equation}
	S_1\to S_2:=\bigvee\{S\in L_{QM}\:;\: S\wedge S_1\leq S_2\}.
\end{equation}
Negation is defined in the standard way as $\neg S:=S\to\bot$.

\section{Discussion}\label{discussie}
The Heyting algebra $L_{QM}$ is not a new logic for quantum mechanics but it was actually already proposed earlier in \cite{CHLS09}. 
Although it was studied there to some extend from a mathematical point of view, no philosophical derivation was given as to why this should be the correct logic for describing quantum systems. 
The only attempt I found at interpreting $L_{QM}$ was the line
\begin{quote}
	``Each element of [$L_{QM}$] corresponds to a ``Bohrified'' proposition, in the sense that to each classical context [$\A\in\mathfrak{A}(\mathcal{C})$] it associates a yes-no question (i.e. an element of the Boolean lattice [$\mathcal{P}(\A)$] of projections in [$\A$]), rather than being a single projection as in standard quantum logic.'' \cite[p. 732]{CHLS09}
\end{quote}
In the present article, an interpretation has been given and I have shown that this interpretation is not only consistent with this logic, but also derived this logic from the interpretation.
A ``Bohrified'' proposition may now be understood as a proposition written as a disjunction of elementary propositions.
However, the elementary propositions as I introduced them (functions of the form $S_{(\A,P)}$) play no role in the article of Caspers et al.
It is an open question if the interpretation given here is the only one consistent with $L_{QM}$.

The classical case results as a special case of the quantum case. 
To see this consider the finest partition of $\Omega$ allowed by the observables $P_{\mathpzc{Obs}}$. 
The set $C_0(P_{\mathpzc{Obs}})$ of all complex valued functions on $P_{\mathpzc{Obs}}$ that vanish at infinity forms an (Abelian) C*-algebra. 
The lattice $L_{QM}$ obtained for this algebra is precisely the lattice $L_{CM2}$. 
The details of this analogy are left to the reader. 

When comparing the lattices $L_{CM}$ and $L_{QM}$ it is rather surprising that, despite their similarity, the first allows a simple modification with the aid of realism, while for the second this is unclear.
At least it is known that the same method for classical mechanics cannot be used for quantum mechanics.
This is a result of the Kochen-Specker theorem which implies that no state space can be defined in which every state dictates the values for all observables in a way consistent with the laws of the theory.
Somewhat peculiarly this hasn't much to do with the specific feature that certain observables cannot be measured simultaneously.
Rather, it has to do with the algebraic structure of the set of observables.
This is explicitly made clear in so-called MKC-models for quantum mechanics \cite{Meyer99}, \cite{Kent99}, \cite{CliftonKent99} and \cite{Hermens11}.
These models can provide an ontology by restricting to a certain subset of observables, most of which still cannot be measured simultaneously.
More popular ontological models like Bohmian mechanics \cite{Bohm52} get by by assuming that every proposition $M_A(\Delta)$ is in fact of the form $M_X(\Delta')$ where $X$ is a special observable denoting the position of some particles (for example, the pointer on a measuring apparatus).
In fact, any advocate of a specific realist interpretation of quantum mechanics will argue that $L_{QM}$ is a very cumbersome logic for reasoning, but many will differ in pointing out which aspect is precisely cumbersome. 

But how cumbersome is the logic in practice?
The lattice $L_{QM}$ is claimed to provide a consistent way to reason about quantum mechanical propositions. 
In most practical cases however, one only considers a finite set of possible measurements and it then seems inappropriate to use propositions that essentially talk about all possible measurements. 
For example, one may ask the question ``what is the appropriate logic if I know I will measure $A$?''
Truth be told, propositions in $L_{QM}$ explicitly refer to the past, but as the future is the future's past one may as well think of the propositions as ones concerning the future.

The natural approach is to say that in the case of an actual measurement a lot of the propositions in $L_{QM}$ may be considered equivalent:
\begin{equation}\label{framework}
	S_1\sim_A S_2\quad\Longleftrightarrow\quad S_1(\mathcal{A})=S_2(\mathcal{A}).
\end{equation}
The set $L_{QM}/\sim_A$ is the Boolean lattice of projection operators in $\mathcal{A}$. 
In other words: given the measurement of $A$, every proposition about outcomes of the measurement becomes decidable. 
So in practice, the intuitionistic logic behaves classically. 
Note that this result is similar to the case of orthodox quantum logic.

In general a measurement of an observable $A$ is not a complete measurement and there are observables of which $A$ is a function. 
So actually a measurement of $A$ only gives certainty about the observables $A'$ with $\mathcal{A}'\subseteq\mathcal{A}$ but not about the observables $A''$ with $\mathcal{A}\subseteq\mathcal{A}''$. 
From this perspective one should introduce the more subtle equivalence relation
\begin{equation}
	S_1\sim'_A S_2\quad\Longleftrightarrow\quad S_1(\mathcal{A}')=S_2(\mathcal{A}')\quad \forall \mathcal{A}'\supset\mathcal{A}.
\end{equation}
If $A$ is a maximal observable this is again the equivalence relation $\sim_A$ but in general it gives rise to a more refined conditional logic
\begin{equation}
	L_{QM}/\sim'_A=\left\{S:\uparrow\mathcal{A}\to\mathcal{P}(\mathcal{C})\:;\:
	\substack{S(\mathcal{A}')\in \mathcal{P}(\mathcal{A}')\text{ and}\\
	S(\mathcal{A}_1)\leq S(\mathcal{A}_2)\text{ whenever } \mathcal{A}_1\subset\mathcal{A}_2}\right\},
\end{equation}
where $\uparrow\mathcal{A}=\{\mathcal{A}'\in\mathfrak{A}(\mathcal{C})\:;\:\mathcal{A}'\supset\mathcal{A}\}$. 
This is again a Heyting algebra (with partial order, join and meet defined analogously to the standard case). 
The decidable elements in this algebra (i.e. the $S$ for which $S\vee\neg S=\top$) are given by the equivalence classes $[S_{(\mathcal{A},P)}]$ with $P\in\mathcal{P}(\mathcal{A})$ i.e. those that correspond with the elementary propositions about measurements of $A$. 
Thus the logic $L_{QM}/\sim'_A$ is the logic in which propositions concerning measurements of $A$ are decidable, but other propositions aren't. 

The decidability of a proposition thus strongly depends on the context in which the proposition is formulated and in general, decidability in one context can't be expected to hold in another context. 
Of course this is already well-known as one may see from Feynman's story on page \pageref{Feyn}.
But with the introduction of $L_{QM}$ his tightrope has become much less of a hand-waving argument.

The equivalence relation $\sim'_A$ was defined by appealing to an uncertainty about what a full measurement of the system would be, but it can also be seen as appealing to a certain notion of locality.
Indeed, given a measurement of $A$, one remains uncertain about any other measurement made on the system possibly at some distance from where the measurement of $A$ was performed. 
Symmetrically, the person performing the distant measurement remains uncertain about the measurement of $A$ (and whether or not it is performed). 
As a consequence, people who are studying the same system at distant locations will in general use different logics describing the system.

As an example consider the standard situation where Alice can choose between two possible measurements $A_1$ and $A_2$ and Bob can choose between two possible measurements $B_1$ and $B_2$. 
Suppose both know from each other that these are the only possible measurements between which they can choose. 
Let $\mathcal{A}_i\mathcal{B}_j$ denote the algebra generated by $\mathcal{A}_i$ and $\mathcal{B}_j$. 
If Alice chooses to measure $A_i$ the appropriate logic for her will be
\begin{equation}
	L_{A_i}=\left\{S:\{\mathcal{A}_i,\mathcal{A}_i\mathcal{B}_1,\mathcal{A}_i\mathcal{B}_2\}\to\mathcal{P}(\mathcal{C})\:;\:\substack{S(\mathcal{A})\in\mathcal{P}(\mathcal{A})\text{ for }\mathcal{A}\in\{\mathcal{A}_i, \mathcal{A}_i\mathcal{B}_1,\mathcal{A}_i\mathcal{B}_2\},\\ S(\mathcal{A}_i)\leq S(\mathcal{A}_i\mathcal{B}_j)\text{ for }j=1,2}\right\}.
\end{equation}
Symmetrically, Bob will use the logic 
\begin{equation}
	L_{B_i}=\left\{S:\{\mathcal{B}_i,\mathcal{A}_1\mathcal{B}_i,\mathcal{A}_2\mathcal{B}_i\}\to\mathcal{P}(\mathcal{C})\:;\:\substack{S(\mathcal{A})\in\mathcal{P}(\mathcal{A})\text{ for }\mathcal{A}\in\{\mathcal{B}_i, \mathcal{A}_1\mathcal{B}_i,\mathcal{A}_2\mathcal{B}_i\},\\ S(\mathcal{B}_i)\leq S(\mathcal{A}_j\mathcal{B}_i)\text{ for }j=1,2}\right\}.
\end{equation}

To properly discuss the connection of these logics with Bell-type inequalities it is necessary to introduce a notion of probability. 
It is beyond the scope of this article to investigate what all consistent possibilities would be for the new quantum logic. 
However, I do want to sketch some ideas.
It does seem natural that for Alice, using the logic $L_{A_i}$, the probabilities provided by the formalism of quantum mechanics that are relevant to her are those that are assigned to the decidable propositions in her logic: those about the possible outcomes for the measurement of $A_i$.
That is, she will assign probabilities to those propositions that are decidable for her.
But for other propositions the situation is not clear-cut.
The proof of theorem \ref{Belllemma} makes more clear what is at stake.
In order for the proof to make sense, one must be able to assign a probability to the proposition $A_1\wedge B_1\wedge (B_2\vee \neg B_2)$.
Although one can identify it with an element in $L_{A_1}$, the probability Alice should assign to it is unknown and may well depend on an underlying ontological model.
However, the structure of $L_{A_1}$ does make clear that it may well be a \emph{lower} value than that what will be assigned to $A_1\wedge B_1$.
Note that this is not necessarily an artifact of $L_{QM}$ being intuitionistic but rather that $\neg B_2$ is not represented in $L_{QM}$ as the negation of $B_2$.
That is, in general one has
\begin{equation}
	S_{(\A,P^\bot)}\leq\neg S_{(\A,P)}
\end{equation}
and no equality. 
Either way, this is the domain where one is to seek a peaceful coexistence between quantum mechanics and some form of locality.

A more recent attempt at finding such a coexistence is delivered by the consistent histories approach \cite{Griffiths11}.
In this approach one only looks at partial logics constructed with equivalence relations of the form (\ref{framework}).
Such a partial logic is called a framework, and one postulates that all reasoning must be performed within one single framework.
It follows roughly from this postulate that no Bell-inequality can be derived since every inequality involves more than one framework.
Now according to Griffiths the choice of a framework is a pure epistemological act; it does not influence the system under consideration.
Therefore, the propositions within such a framework are purely epistemic too, for if they were ontological, a change of framework would influence the system.
For example, a sentence about a property of the system could shift from being true to being meaningless.
It thus seems to me that the consistent histories approach is not capable of properly describing an ontology for quantum systems, which makes the entire discussion of locality (which is, to be sure, an ontological concept) quite meaningless.
Indeed, the logic suggested by Griffiths is quite reminiscent of what I have done here, with the difference that $L_{QM}$ acts on a framework-transcending level; every sentence in $L_{QM}$ is a disjunction of sentences, each formulated within a single framework.

In conclusion, I don't believe that the problems in the foundations of quantum mechanics could vanish by introducing the `correct' logic.
Furthermore, it is not my opinion that $L_{QM}$ should be conceived as the correct logic.
I do think that logic can play an important role in carefully investigating the philosophical problems we face in quantum mechanics.
A careful distinction between empirical and ontological assumptions is mandatory for this, and I think the logic $L_{QM}$ can help in making this distinction more clear.

\section{Acknowledgements}
I could not have written this article in this present form without the positive influences of R. Griffiths, M. P. Seevinck, S. Wolters and one anonymous referee.
\bibliographystyle{newapa}    
\bibliography{referenties}

\end{document}